\documentclass[12pt,a4paper]{article}

\title{}
\author{}
\date{}

\usepackage{amsmath,amsfonts,amssymb,amsthm}
\usepackage{latexsym}

\makeatletter
\@addtoreset{equation}{section}
\makeatother

\newcommand{\be}{\begin{equation}}
\newcommand{\ee}{\end{equation}}
\newcommand{\ben}{\begin{equation}}
\newcommand{\een}{\end{equation}}
\newcommand{\bea}{\setlength\arraycolsep{2pt} \begin{eqnarray}}
\newcommand{\eea}{\end{eqnarray}}
\newcommand{\nnr}{\nonumber \\}

\newcommand{\eq}[1]{(\ref{#1})}

\newcommand{\fr}{\frac}
\newcommand{\tf}{\tfrac}

\newcommand{\wtd}{\widetilde}

\newcommand{\df}{\textrm{d}}
\newcommand{\expe}[1]{\textrm{e}^{#1}}
\newcommand{\pd}{\partial}
\newcommand{\sr}{\sqrt}

\newcommand{\gep}{\epsilon}

\newcommand{\gvf}{\varphi}

\newcommand{\im}{\textrm{i}}

\newcommand{\cL}{\mathcal{L}}

\newtheorem{proposition}{Proposition}
\newtheorem{lemma}[proposition]{Lemma}
\newtheorem{corollary}[proposition]{Corollary}

\theoremstyle{definition}

\begin{document}

\thispagestyle{empty}

\begin{flushright}
CCQCN-2015-112
\end{flushright}
\vspace*{100pt}
\begin{center}
\textbf{\Large{Higher-dimensional lifts of Killing--Yano forms with torsion}}\\
\vspace{50pt}
\large{David D. K. Chow}
\end{center}

\begin{center}
\textit{Crete Center for Theoretical Physics and\\
Crete Center for Quantum Complexity and Nanotechnology,\\
Department of Physics, University of Crete, 71003 Heraklion, Greece}\\
{\tt dchow@physics.uoc.gr}\\
 \vspace{30pt}
 {\bf Abstract\\}
 \end{center}
Using a Kaluza--Klein-type lift, it is shown how Killing--Yano forms with torsion can remain symmetries of a higher-dimensional geometry, subject to an algebraic condition between the Kaluza--Klein field strength and the Killing--Yano form.  The lift condition's significance is highlighted, and is satisfied by examples of black holes in supergravity.
\newpage


\section{Introduction}	


Killing--Yano (KY) $p$-forms \cite{Yano} generalize Killing vectors to higher-rank antisymmetric tensors $Y_{a_1 \ldots a_p} = Y_{[a_1 \ldots a_p]}$ that satisfy
\be
\nabla_a Y_{b_1 \ldots b_p} = \nabla_{[a} Y_{b_1 \ldots b_p]} .
\ee
The Kerr--Newman black hole spacetime admits a KY 2-form \cite{Floyd, Penrose:1973um}, which underlies many of the solution's remarkable properties.  Higher-dimensional generalizations in Einstein gravity possess additional KY $p$-forms and retain many of the remarkable properties of the 4-dimensional Kerr spacetime; see \cite{Kubiznak:2008qp, Yasui:2011pr} for reviews.  Like Killing vectors, KY forms give symmetries.  However, they are ``hidden symmetries'' of phase space, rather than configuration space; see \cite{Cariglia:2014ysa} for a review.  Some particular consequences of KY forms are: constants of motion for charged particle motion \cite{Hughston:1972qf}; the existence of an operator commuting with the Dirac operator \cite{Carter:1979fe}; and enhanced worldline supersymmetry of a spinning particle \cite{Gibbons:1993ap}.

There are generalizations of the Kerr--Newman black hole that are charged, rotating black hole solutions of supergravities.  These theories contain the 3-form Kalb--Ramond field strength $H$ of string theory.  Some of the known solutions admit generalizations of KY forms in which the connection is modified to include a torsion \cite{YanoBochner}, identified with the 3-form $H$, for example in \cite{Wu:2009cn, Kubiznak:2009qi, Houri:2010fr, Chow:2013gba}.  Classifying supersymmetric solutions of supergravity leads to $G$-structures, on which KY forms with torsion also appear naturally \cite{Papadopoulos:2007gf, Papadopoulos:2011cz}.  Separately, KY 3-forms arise in 11-dimensional supergravity reduced on $S^7$; for squashed $S^7$ a KY 3-form constructed from a Killing spinor gives a Ricci-flattening torsion \cite{Englert:1983qe}, while the 70 KY 3-forms of round $S^7$ are associated with 35 massless and 35 massive pseudoscalars \cite{Duff:1983gh, Duff:1986hr}.  See \cite{Santillan:2011sh} for a review of KY forms in supersymmetric theories and $G$-structures.

One reason for studying black holes in supergravity, rather than more general theories involving gravity, is that it seems to be easier to find exact solutions that are charged, rotating black holes.  A charged, rotating exact black hole solution of higher-dimensional Einstein--Maxwell theory is not known, whereas examples are known explicitly in supergravity.  Although there are solution generating techniques arising from string theory dualities, these do not explain the construction of asymptotically anti-de Sitter (AdS) black holes in gauged supergravity.  However, behind all of these known solutions are Killing tensors.  In general, these are symmetric Killing--St\"{a}ckel tensors rather than antisymmetric KY forms with torsion; see e.g.\ \cite{Chow:2008fe}.  Killing tensors may provide guidance for finding new solutions.  Whereas a general charged, rotating black hole in 5-dimensional minimal gauged supergravity was originally found using inspired guesswork \cite{Chong:2005hr}, it can be derived assuming a certain type of Killing tensor \cite{Ahmedov:2009ab}.  Some general classes of spacetimes admitting KY forms with torsion have been found \cite{Houri:2012eq, Houri:2012su}.  For exact solutions outside supergravity, KY forms with torsion have inspired the generalization of the Wahlquist perfect fluid solution to higher dimensions \cite{Hinoue:2014zta}.

A second reason for studying black holes in supergravity is because string theory may provide a framework for a microscopic understanding of black hole entropy.  The known solutions are obtained by solving lower-dimensional reductions rather than string theory or M-theory in 10 or 11 dimensions, but microscopic countings use higher-dimensional objects.  The structure of a higher-dimensional solution may differ from its lower-dimensional interpretation, for example singularities might be resolved.  Higher dimensions are used for solution generating techniques in string theory; in this context, there has been recent study of KY forms in supergravity black holes \cite{Chervonyi:2015ima}.  It is therefore of interest to study solutions from a higher-dimensional perspective.

In this paper, motivated by black hole examples in supergravity, we consider higher-dimensional lifts of KY forms with 3-form torsion included.  We consider a specific Kaluza--Klein lift of the metric and 3-form and prove general results about whether a lower-dimensional symmetry is also a higher-dimensional symmetry.  There is an algebraic condition between the Kaluza--Klein field strength and the KY forms with torsion.  We highlight the significance of this lift condition, as it appears in various consequences of KY forms mentioned above.  Although we prove some general results about lifting Killing tensors, without assuming any field equations, our Kaluza--Klein ansatz is motivated by known examples in supergravity, for which the lift condition holds.

Note that other higher-dimensional lifts of Killing tensors have also been studied; in particular, two types of lifts arising from the work of Eisenhart \cite{Eisenhart}.  One Eisenhart lift is a warped product lift with one extra dimension, leading to studies of more general warped products \cite{Benn, Krtous:2015ona}.  The other Eisenhart lift involves two extra dimensions, going from Euclidean signature to Lorentzian signature; see e.g.\ \cite{Cariglia:2012fi}.  Conformal Killing--Yano forms on metric cones have been studied in \cite{Semmelmann}.

The outline of this paper is as follows.  In Section 2, we recall the basic definitions of Killing tensors and prove general results for lifting Killing tensors to higher dimensions using a Kaluza--Klein-type ansatz.  We highlight the lift condition for a Killing--Yano form with torsion, discuss its solution, and note its appearances in other contexts.  In Section 3, we apply these general results to examples of black holes in string theory.  We conclude in Section 4.


\section{Killing tensors}


In this section, we study general properties of Killing--Yano forms with torsion,  both their geometric properties and physical consequences.  These results are independent of any field equations or details of the spacetime.


\subsection{Definitions}


We summarize here the basic definitions of the (conformal) Killing tensors that we use.

Firstly, there are antisymmetric conformal Killing--Yano $p$-forms \cite{Tachibana, Kashiwada}.  We also allow a torsion that is derived from a 3-form $T$, so that we have a covariant derivative $\nabla^T$ with connection $\Gamma{^T}{^a}{_{b c}} = \Gamma{^a}{_{b c}} + \tf{1}{2} T{^a}{_{b c}}$, where $\Gamma{^a}{_{b c}}$ is the Levi-Civita connection and $T_{a b c} = T_{[a b c]}$, which is relevant for $p \geq 2$.  In $D$ spacetime dimensions, a conformal Killing--Yano $p$-form with torsion (CKYT $p$-form\footnote{I use the more descriptive term ``with torsion'' following \cite{Howe:1996kj}, rather than ``generalized'', ``modified'' or ``pseudo'' in other literature.}) $Y_{a_1 \ldots a_p} = Y_{[a_1 \ldots a_p]}$ satisfies
\begin{align}
\nabla{^T}{_a} Y_{b_1 \ldots b_p} & = \nabla{^T}{_{[a}} Y_{b_1 \ldots b_p]} + p g_{a [b_1} \widehat{Y}_{b_2 \ldots b_p]} , & \widehat{Y}_{b_2 \ldots b_p} & = \fr{1}{D - p + 1} \nabla{^T}{_c} Y{^c}{_{b_2 \ldots b_p}} .
\end{align}
In other words, we have a $p$-form $Y$ whose derivative can be decomposed into an exterior derivative plus a divergence, where the derivatives include torsion.  If the divergence part is missing, i.e.\ $\widehat{Y}_{b_2 \ldots b_p} = 0$, then we have a Killing--Yano $p$-form with torsion (KYT $p$-form),
\be
\label{KYT}
\nabla{^T}{_a} Y_{b_1 \ldots b_p} = \nabla{^T}{_{[a}} Y_{b_1 \ldots b_p]} .
\ee
An equivalent way of expressing the KYT equation is $\nabla{^T}{_{(a}} Y_{b_1 ) \ldots b_p} = 0$.  If instead the exterior part is missing, i.e.\ $\nabla{^T}{_{[a}} Y_{b_1 \ldots b_p]} = 0$, then we have a closed conformal Killing--Yano $p$-form with torsion (CCKYT $p$-form),
\begin{align}
\nabla{^T}{_a} Y_{b_1 \ldots b_p} & = p g_{a [b_1} \widehat{Y}_{b_2 \ldots b_p]} , & \widehat{Y}_{b_2 \ldots b_p} & = \fr{1}{D - p + 1} \nabla{^T}{_c} Y{^c}{_{b_2 \ldots b_p}} .
\end{align}
The Hodge dual of a KYT $p$-form is a CCKYT $(D - p)$-form, and vice versa \cite{Kubiznak:2009qi}.

Secondly, there are symmetric Killing tensors.  A rank-$p$ Killing--St\"{a}ckel (KS) tensor is a symmetric tensor $K_{a_1 \ldots a_p} = K_{(a_1 \ldots a_p)}$ that satisfies
\be
\label{KS}
\nabla_{(a} K_{b_1 \ldots b_p)} = 0 .
\ee
From a KYT $p$-form $Y_{a_1 \ldots a_p}$, we can construct a rank-2 KS tensor
\be
\label{Ysquared}
K_{a b} = \tf{1}{(p - 1)!} Y{^{c_1 \ldots c_{p - 1}}}{_a} Y_{c_1 \ldots c_{p - 1} b} .
\ee


\subsection{General results}


We now prove some general results concerning Killing tensors and a particular type of Kaluza--Klein lift ansatz.  Although the ansatz is naturally motivated by torus reductions in Kaluza--Klein theory, the results that we prove here are geometric results independent of any particular theory, i.e.\ do not use field equations.  We show that under certain conditions, (CC)KYT forms in lower dimensions lift to (CC)KYT forms in higher dimensions.  A key role is played by the KYT lift condition
\be
\label{KYTliftcond}
F{^b}{_{[a_1}} Y_{a_2 \ldots a_p] b} = 0 .
\ee
This algebraic condition relates a Kaluza--Klein gauge field strength $F$ and a KYT form $Y$, and has further consequences that we discuss later.

We consider a $D$-dimensional metric $\df s^2$ with a gauge field $A$ and a 3-form $H$.  For the results here, we define its Kaluza--Klein lift as a metric $\df \overline{s}^2$ and a 3-form $\overline{H}$ given by the ansatz
\begin{align}
\label{KKansatz}
\df \overline{s}^2 & = \df s^2 + (\df z + A)^2 , & \overline{H} & = H + F \wedge (\df z + A) ,
\end{align}
where $F = \df A$.  The $D$-dimensional metric is $\df s^2 = g_{a b} \, \df x^a \, \df x^b$, and the $(D + 1)$-dimensional metric is $\df \overline{s}^2 = g_{A B} \, \df x^A \, \df x^B$, with $\{ x^A \} = \{ x^a, z \}$.

If we were specializing to a Kaluza--Klein reduction of a $(D+1)$-dimensional bosonic string theory, then this is not the most general ansatz for $\df \overline{s}^2$ and $\overline{H}$, as each of these reduces to give the same gauge field.  Also, no scalars appear in the metric ansatz, through the choices of conformal frames for both metrics, and also because one scalar is consistently truncated.  However, this ansatz is sufficient for the examples that we shall consider later, and for purely geometric results can be used for any $D$-dimensional geometry.

We shall need the following result from the connections.
\begin{lemma}
\label{DeltaGamma}
The non-vanishing components of the difference of connections including torsion $(\Delta \Gamma^H){^A}{_{B C}} = \overline{\Gamma}{^{\overline{H}}}{^A}{_{B C}} - \Gamma{^H}{^A}{_{B C}} $ are
\be
(\Delta \Gamma^H){^a}{_{b c}} = A_b F{_c}{^a} .
\ee
\end{lemma}

\begin{proof}
The metric inverses are related by
\be
\overline{g}^{A B} \, \pd_A \, \pd_B = g^{a b} (\pd_a - A_a \, \pd_z) (\pd_b - A_b \, \pd_z) + \pd_z^2 .
\ee
For the lower-dimensional metric, the Levi-Civita connection is $\Gamma{^a}{_{b c}} = \tf{1}{2} g^{a d} (\pd_b g_{d c} + \pd_c g_{b d} - \pd_d g_{b c})$.  The total connection including torsion is $\Gamma{^H}{^a}{_{b c}} = \Gamma{^a}{_{b c}} + \tf{1}{2} H{^a}{_{b c}}$.  Similarly, the total connection for the higher-dimensional metric is $\overline{\Gamma}{^{\overline{H}}}{^A}{_{B C}} = \overline{\Gamma}{^A}{_{B C}} + \tf{1}{2} \overline{H}{^A}{_{B C}}$.  The result follows by computation.
\end{proof}

Our main result concerns lifting KYT $p$-forms.
\begin{proposition}
\label{KYTlift}
Suppose that the metric $\df s^2$ admits a KYT $p$-form $Y_{a_1 \ldots a_p}$ with 3-form torsion $H$.  Suppose also that the KYT lift condition $F{^b}{_{[a_1}} Y_{a_2 \ldots a_p] b} = 0$ holds.  Then $Y_{a_1 \ldots a_p}$ is a KYT $p$-form for the metric $\df \overline{s}^2$, with torsion $\overline{H}$.
\end{proposition}

\begin{proof}
From the difference of the KYT equations \eq{KYT} for both metrics, we have only the terms involving the connection contracted with the KYT $p$-forms.  Using Lemma \ref{DeltaGamma}, this gives the condition $F{^b}{_{[a_1}} Y_{a_2 \ldots a_p] b} = 0$.
\end{proof}

There is an analogous result for CCKYT $p$-forms.
\begin{corollary}
Suppose that the metric $\df s^2$ admits a CCKYT $p$-form $\wtd{Y}$ with 3-form torsion $H$.  Suppose also that
\be
\label{CCKYTliftcond}
\gep_{c b_1 \ldots b_{D - p} a_1 \ldots a_{p - 1}} F{^c}{_{a_p}} \wtd{Y}^{a_1 \ldots a_p} = 0 .
\ee
Then $\wtd{Y} \wedge (\df z + A)$ is a CCKYT $(p + 1)$-form for the metric $\df \overline{s}^2$, with torsion $\overline{H}$.
\end{corollary}

\begin{proof}
KYT forms and CCKYT forms are Hodge duals to each other \cite{Kubiznak:2009qi}, so if $\wtd{Y}$ is a CCKYT $p$-form, then $Y = \star \wtd{Y}$ is a KYT $(D - p)$-form.  Rewrite the condition $F{^b}{_{[a_1}} Y_{a_2 \ldots a_{D - p}] b} = 0$ of Proposition \ref{KYTlift} in terms of $\wtd{Y}$ through
\be
Y_{a_1 \ldots a_{D - p}} = \tfrac{1}{p!} \gep_{a_1 \ldots a_{D - p} b_1 \ldots b_p} \wtd{Y}^{b_1 \ldots b_p}
\ee
and use the Schouten identity
\be
F{^c}{_{[a_1}} \gep_{a_2 \ldots a_{D - p} c b_1 \ldots b_p]} \wtd{Y}^{b_1 \ldots b_p} = 0 .
\ee
This gives the condition \eq{CCKYTliftcond}.
\end{proof}

If $\wtd{Y}_1$ and $\wtd{Y}_2$ are CCKYT forms, then $\wtd{Y}_1 \wedge \wtd{Y}_2$ is also a CCKYT form \cite{Kubiznak:2009qi}, which we show can be lifted if both $\wtd{Y}_1$ and $\wtd{Y}_2$ can be lifted.
\begin{corollary}
\label{PCCKYTlift}
Suppose that the metric $\df s^2$ admits a CCKYT $p$-form $\wtd{Y}_1$ and a CCKYT $q$-form $\wtd{Y}_2$ with 3-form torsion $H$, and that they satisfy the conditions to lift to CCKYT forms $\wtd{Y}_1 \wedge (\df z + A)$ and $\wtd{Y}_2 \wedge (\df z + A)$ on $\df \overline{s}^2$.  Then $\wtd{Y}_1 \wedge \wtd{Y}_2$ also satisfies the condition to lift to a $(p + q + 1)$-form $\wtd{Y}_1 \wedge \wtd{Y}_2 \wedge (\df z + A)$ on $\df \overline{s}^2$.
\end{corollary}

\begin{proof}
By the assumption that $\wtd{Y}_1$ lifts, we have
\be
\gep_{c b_1 \ldots b_{D - p - q} a_{p + 1} \ldots a_{p + q} a_1 \ldots a_{p - 1}} F{^c}{_{a_p}} (\wtd{Y}{_1})^{a_1 \ldots a_p} = 0 .
\ee
Contract with $(\wtd{Y}_2)^{a_{p + 1} \ldots a_{p + q}}$ and antisymmetrize the $a$ indices to obtain
\be
\gep_{c b_1 \ldots b_{D - p - q} a_1 \ldots a_{p + q - 1}} F{^c}{_{a_{p + q}}} (\wtd{Y}_1 \wedge \wtd{Y}_2)^{a_1 \ldots a_{p + q}} = 0 ,
\ee
which is the condition for $\wtd{Y}_1 \wedge \wtd{Y}_2$ to lift.
\end{proof}

Symmetric KS tensors can be lifted, again subject to an algebraic condition.
\begin{proposition}
Suppose that the metric $\df s^2$ admits a rank-$p$ KS tensor $K_{a_1 \ldots a_p}$.  Suppose also that the KS lift condition
\be
\label{FK}
F{^b}{_{(a_1}} K_{a_2 \ldots a_p) b} = 0
\ee
holds.  Then $K_{a_1 \ldots a_p}$ is a rank-$p$ KS tensor for the metric $\df \overline{s}^2$.
\end{proposition}

\begin{proof}
From the difference of the KS equations \eq{KS} for both metrics, we have only the terms involving the connection contracted with the KS tensors.  Using Lemma \ref{DeltaGamma}, this gives the condition $F{^b}{_{(a_1}} K_{a_2 \ldots a_p) b} = 0$.
\end{proof}

Constants of motion are independent if they are in involution, i.e.\ Poisson commute.  For the constants of motion $K_{a_1 \ldots a_p} P^{a_1} \ldots P^{a_p}$ associated with rank-$p$ KS tensors, this is equivalent to the vanishing of the Schouten--Nijenhuis brackets of the KS tensors.  The lift preserves the vanishing of these brackets.  Note that Poisson brackets are defined for quantities on phase space, rather than configuration space, so the corresponding Schouten--Nijenhuis brackets involve partial derivatives.  These can be written in a manifestly covariant way as standard covariant derivatives, rather than torsion-modified covariant derivatives as suggested in \cite{Houri:2010fr}.

CKYT forms are associated to first-order symmetry operators for torsion-modified Dirac equations.  However, when torsion is included, in order to construct a symmetry operator, there are further anomaly terms that must vanish \cite{Houri:2010qc, Kubiznak:2010ig}.  It is straightforward to see that the lift preserves the vanishing of anomalies, and so the symmetry operators can be lifted.


\subsection{Solution of lift condition}
\label{liftsol}


A significant application of Corollary \ref{PCCKYTlift} is to a non-degenerate CCKYT 2-form $\wtd{Y}$, also known as a principal CKYT form.  By taking multiple powers, $\wtd{Y}^{(j)} = \fr{1}{j!} \wtd{Y} ^{\wedge j}$, for $j = 1, 2, \ldots, \lfloor D/2 \rfloor$, we obtain a tower of CCKYT forms, which can be Hodge dualized to a tower of KYT forms.  To show that all $\wtd{Y}^{(j)}$ satisfy the CCKYT lift condition, it suffices to check $\wtd{Y}$ only.

To give the explicit solution for the lift condition of a non-generate CCKYT 2-form, which is Hodge dual to a KYT $(D - 2)$-form, we work in a Darboux basis for $\wtd{Y}$.  We define $\varepsilon = 0, 1$ by $D = 2 n + \varepsilon$, according to whether the dimension is even or odd.  There are pairs of vielbeins $e^\mu$, $e^{\hat{\mu}}$, $\mu = 1, \ldots, n$, with an extra unpaired vielbein $e^0$ present only if $D$ is odd.  Non-degeneracy means that the eigenvalues of the endomorphism $\wtd{Y}{^a}{_b}$ are functionally independent in some domain.  At a generic point, they must come in non-zero and distinct pairs, plus a single zero eigenvalue in odd dimensions.  The CCKYT 2-form takes the form
\be
\label{Ytilde}
\wtd{Y} = \sum_{\mu = 1}^n x_\mu e^\mu \wedge e^{\hat{\mu}} ,
\ee
where $x_\mu$ are constants at that point, which are distinct in Euclidean signature.  The CCKYT lift condition \eq{CCKYTliftcond} implies that $F$ is a linear combination (at that point) of $e^\mu \wedge e^{\hat{\mu}}$, $\mu = 1, \ldots, n$.  In 3 dimensions, $F$ is a multiple of $\wtd{Y}$.  In 4 dimensions, $F$ is a linear combination of $Y$ and $\wtd{Y}$, as discussed in \cite{Carter:1987id}, where the solution is expressed in terms of a complex self-dual 2-form.

For a KYT 2-form $Y$, in a Darboux basis for $Y$ itself, we have
\be
Y = \sum_{\mu = 1}^n x_\mu e^\mu \wedge e^{\hat{\mu}} .
\ee
Similarly, we find that $F$ is a linear combination of $e^\mu \wedge e^{\hat{\mu}}$, $\mu = 1, \ldots, n$.

For a KYT $p$-form with $3 \leq p \leq D - 3$, the lift condition has many more components than $F$, and generically implies that $F = 0$.  Only in special cases, such as when the KYT $p$-form is induced by a CCKYT 2-form, can there be a non-trivial $F$.


\subsection{Applications of lift condition}
\label{lift}


The KYT lift condition \eq{KYTliftcond} also appears in several other contexts, which we summarize here.  Even if the particular Kaluza--Klein lift of a solution is not physically motivated by the theories, the properties below will still apply.


\subsubsection{Charged particle motion}


A particle with momentum $P_a$ charged under the gauge field $A$, with electric charge $e$, is subject to the Lorentz force law $P^b \nabla_b P^a = e F{^a}{_b} P^b$.  For a rank-2 KS tensor $K_{a b}$, we have
\be
P^a \nabla_a (K_{b c} P^b P^c) = 2 e K_{b c} F{^b}{_a} P^a P^c .
\ee
Although $K_{a b} P^a P^b$ is conserved along the paths of uncharged particles, it is not in general conserved along the paths of charged particles.  However, suppose in addition that the KS tensor is the square of a KYT $p$-form, as in \eq{Ysquared}, and that the lift condition \eq{KYTliftcond} is satisfied.  Then we have
\begin{align}
F{^b}{_a} K_{c b} & = F{^b}{_a} Y{^{d_1 \ldots d_{p - 1}}}{_c} Y_{d_1 \ldots d_{p - 1} b} \nnr
& = (p - 1) F{^b}{_{d_{p - 1}}} Y{^{d_1 \ldots d_{p - 1}}}{_c} Y_{d_1 \ldots d_{p - 2} a b} \nnr
& = (p - 1) F{^{d_{p - 1}}}{_b} Y{^{d_1 \ldots d_{p - 2} b}}{_c} Y_{d_1 \ldots d_{p - 2} a d_{p - 1}} \nnr
& = F_{c b} Y{^{d_1 \ldots d_{p - 2} b}}{_{d_{p - 1}}} Y_{d_1 \ldots d_{p - 2} a}{^{d_{p - 1}}} \nnr
& = - F{^b}{_c} K_{a b} ,
\end{align}
and so
\be
\label{FK2}
F{^b}{_{(a}} K_{c) b} = 0 .
\ee
It follows that $K_{a b} P^a P^b$ is conserved along the paths of charged particles.  For a KY 2-form, the relevance of the lift condition \eq{KYTliftcond} for charged particle motion was shown in \cite{Hughston:1972qf}, and a slightly lengthier derivation for $p$-forms was given in \cite{Acik:2008wv}.  Note that \eq{FK2} is the lift condition \eq{FK} for a rank-2 KS tensor.


\subsubsection{Dirac operator}


Consider the Dirac equation for a charged spinor in a curved background with an electromagnetic field.  The Dirac operator is $\mathcal D = \gamma^a (\nabla_a - \im e A_a)$.  If there is a KY 2-form $Y_{a b}$ that satisfies the constraint
\be
\label{Diraccondition}
F{^c}{_{[a}} Y_{b] c} = 0 ,
\ee
then one can construct a gauge-invariant operator that commutes with $\mathcal D$ \cite{Carter:1979fe}, giving a good quantum number.  For an uncharged spinor, the construction has been generalized to KY $p$-forms \cite{Cariglia:2003kf}, and further including torsion to KYT $p$-forms \cite{Houri:2010fr}.  For a charged spinor, the construction has been generalized to KY $p$-forms \cite{Acik:2008wv}, and to KYT $p$-forms \cite{Kubiznak:2010ig}.  For convenience, we give here a self-contained demonstration of the result to include charge and arbitrary rank, but without torsion, largely following \cite{Houri:2010fr}.

\begin{proposition}
\label{KYop}
Let $Y_{a_1 \ldots a_p}$ be a KY $p$-form, and define the operator
\be
K_Y = \gamma^{b_1 \ldots b_{p - 1}} Y{^a}{_{b_1 \ldots b_{p - 1}}} (\nabla_a - \im e A_a) + \fr{1}{2 (p + 1)} \gamma^{b_1 \ldots b_{p + 1}} \nabla_{b_1} Y_{b_2 \ldots b_{p + 1}} .
\ee
Suppose also that
\be
F{^b}{_{[a_1}} Y_{a_2 \ldots a_p] b} = 0 .
\ee
Then $K_Y$ graded anti-commutes with the Dirac operator $\mathcal D$:
\be
\mathcal D K_Y + (-1)^p K_Y \mathcal D = 0 .
\ee
\end{proposition}

\begin{proof}
The case $e = 0$ is already known explicitly, and the terms proportional to $e^2$ are trivial.  We therefore only need to consider the terms linear in $e$, checking that
\begin{align}
& \gamma^c \gamma^{b_1 \ldots b_{p - 1}} [A_c Y{^a}{_{b_1 \ldots b_{p - 1}}} \nabla_a \psi + \nabla_c (Y{^a}{_{b_1 \ldots b_{p - 1}}} A_a \psi)] \nnr
& + (-1)^p \gamma^{b_1 \ldots b_{p - 1}} \gamma^c Y{^a}{_{b_1 \ldots b_{p - 1}}} [\nabla_a (A_c \psi) + A_a \nabla_c \psi] \nnr
& + \fr{1}{2 (p + 1)} [\gamma^c \gamma^{b_1 \ldots b_{p + 1}} + (-1)^p \gamma^{b_1 \ldots b_{p + 1}} \gamma^c ] A_c (\nabla_{b_1} Y_{b_2 \ldots b_{p + 1}}) \psi = 0 .
\end{align}
Note the gamma matrix identities
\begin{align}
\gamma^c \gamma^{b_1 \ldots b_{p - 1}} & = \gamma^{c b_1 \ldots b_{p - 1}} + (p - 1) g^{c [b_1} \gamma^{b_2 \ldots b_{p - 1}]} , \nnr
\gamma^{b_1 \ldots b_{p - 1}} \gamma^c & = (-1)^p [- \gamma^{c b_1 \ldots b_{p - 1}} + (p - 1) g^{c [b_1} \gamma^{b_2 \ldots b_{p - 1}]}] .
\end{align}
Consider separately the terms involving derivatives of $\psi$, $Y$ and $A$.  Firstly, we have
\be
2 (p - 1) g^{c [b_1} \gamma^{b_2 \ldots b_{p - 1}]} Y{^a}{_{b_1 \ldots b_{p - 1}}} (A_c \nabla_a + A_a \nabla_c) \psi = 0 ,
\ee
since $Y{^{(a c)}}{_{b_2 \ldots b_{p - 1}}} = 0$.  Secondly, we have
\begin{align}
& [\gamma^{c b_1 \ldots b_{p - 1}} + (p - 1) g^{c [b_1} \gamma^{b_2 \ldots b_{p - 1}]}] A_a (\nabla_c Y{^a}{_{b_1 \ldots b_{p - 1}}}) \psi + g^{c [b_1} \gamma^{b_2 b_3 \ldots b_{p + 1}]} A_c (\nabla_{b_1} Y_{b_2 b_3 \ldots b_{p + 1}}) \psi \nnr
& = \gamma^{c b_1 \ldots b_{p - 1}} A^a (\nabla_c Y_{a b_1 \ldots b_{p - 1}}) \psi + \gamma^{c b_1 \ldots b_{p - 1}} A^a (\nabla_{[a} Y_{c b_1 \ldots b_{p - 1}]}) \psi \nnr
& = 0 ,
\end{align}
in the first step using $\nabla_c Y{^{a c}}{_{b_2 \ldots b_{p - 1}}} = 0$ in the first term and relabelling in the second term.  Finally, we require
\begin{align}
0 & = \gamma^{c b_1 \ldots b_{p - 1}} Y{^a}{_{b_1 \ldots b_{p - 1}}} F_{c a} \psi + g^{c [b_1} \gamma^{b_2 \ldots b_{p - 1}]} Y{^a}{_{b_1 \ldots b_{p - 1}}} (\nabla_c A_a + \nabla_a A_c) \psi \nnr
& = \gamma^{c b_1 \ldots b_{p - 1}} Y{^a}{_{[b_1 \ldots b_{p - 1}}} F_{c] a} ,
\end{align}
since $Y{^{(a c)}}{_{b_2 \ldots b_{p - 1}}} = 0$.  This gives the condition \eq{Diraccondition}.
\end{proof}
The rather lengthier generalization to include torsion has been given in \cite{Kubiznak:2010ig}.


\subsubsection{Worldline supersymmetry}


For a pseudo-classical description of a spinning particle, there is a connection between between KY forms and enhanced worldline supersymmetry, as a superinvariant can be constructed from a KY form.  Originally shown for an uncharged spinning particle in the case of a KY 2-form \cite{Gibbons:1993ap}, the connection has been generalized for charged particles to arbitrary rank $p$-forms without torsion \cite{Tanimoto:1995np}, and 2-forms with torsion \cite{Rietdijk:1995ye}.  In these two generalizations, the lift condition \eq{KYTliftcond} is necessary.  The case of arbitrary rank with torsion but for an uncharged particle, which places no constraint between the electromagnetic field and the KYT $p$-form, was considered in \cite{Kubiznak:2009qi}.  We expect that the lift condition \eq{KYTliftcond} is required for constructing a superinvariant in the general case of a charged particle and an arbitrary rank KYT $p$-form.  

The lift condition \eq{KYTliftcond} also appears as a condition for the invariance of a general 2-derivative worldline action with $\mathcal N = 1$ supersymmetry \cite{Papadopoulos:2011cz}.  Suppose that the action is written in terms of $\mathcal N = 1$ superfields $X^a$, which are maps from $\mathcal N = 1$ superspace $\Xi^{1 | 1}$ to the target space.  If we perform a transformation $\delta X{^a} = Y{^a}{_{b_1 \ldots b_{p - 1}}} D X^{b_1} \ldots D X^{b_{p - 1}}$ for a target space $p$-form $Y$, then $Y$ must be a KYT $p$-form and satisfies the lift condition \eq{KYTliftcond}.


\section{Black hole examples}


The black holes in supergravity known to admit KYT forms are \cite{Chow:2013gba, Mei:2007bn, Chow:2008ip, Chow:2007ts, Chow:2008fe}.  These all satisfy the Killing tensor lift conditions.  We give several explicit examples here.


\subsection{Charged Kerr--NUT solution in all dimensions}


Consider a bosonic string theory in arbitrary spacetime dimensions $D \geq 4$.  Two noteworthy conformal frames are Einstein frame and string frame \cite{Callan:1985ia}.  The $D$-dimensional theory we consider also has a gauge field, and has Einstein frame Lagrangian
\be
\label{stringEinstein}
\cL_D = R \star 1 - \tf{1}{2} \star \df \gvf \wedge \df \gvf - X^{-2} \star F \wedge F - \tf{1}{2} X^{-4} \star H \wedge H ,
\ee
where the 3-form field strength is $H = \df B - A \wedge F$ and $X = \expe{-\gvf/\sr{2 (D - 2)}}$.  The string frame metric $\df s^2$ and the Einstein frame metric $\df s_{\textrm{E}}^2$ are related by
\be
\df s^2 = X^2 \, \df s_{\textrm{E}}^2 .
\ee
The $D$-dimensional string frame Lagrangian is
\be
\cL_D = X^{-(D-2)} [R \star 1 + \tf{1}{2} (D - 2) \star \df \gvf \wedge \df \gvf - \star F \wedge F - \tf{1}{2} \star H \wedge H] .
\ee
We lift to a $(D + 1)$-dimensional theory using the ansatz for the $(D + 1)$-dimensional fields
\begin{align}
\df \overline{s}^2 & = \df s^2 + (\df z + A)^2 , & \overline{H} = H + F \wedge (\df z + A) ,
\end{align}
where the $\df \overline{s}^2$ is the $(D + 1)$-dimensional string frame metric, $\df s^2$ is the $D$-dimensional string frame metric, and $\overline{H} = \df \overline{B}$ is the $(D + 1)$-dimensional 3-form field strength.  This motivates the ansatz \eq{KKansatz} used previously.  We now work with the $(D + 1)$-dimensional theory and remove the bars on the fields.  The $(D+1)$-dimensional string frame Lagrangian is
\be
\cL_{D + 1} = X^{-(D-2)} [R \star 1 + \tf{1}{2} (D - 2) \star \df \gvf \wedge \df \gvf - \tf{1}{2} \star H \wedge H] .
\ee
The $(D + 1)$-dimensional string frame metric $\df s^2$ and Einstein frame metric $\df s_{\textrm{E}}^2$ are related by
\be
\df s^2 = X^{2 (D - 2)/(D - 1)} \, \df s_{\textrm{E}}^2 .
\ee
The $(D + 1)$-dimensional Einstein frame Lagrangian is
\be
\cL_{D + 1} = R \star 1 - \fr{D - 2}{2 (D - 1)} \star \df \gvf \wedge \df \gvf - \fr{1}{2} X^{-4 (D - 2)/(D - 1)} \star H \wedge H .
\ee

For these theories, there is a charged Kerr--NUT solution in arbitrary dimensions $D \geq 4$ \cite{Chow:2008fe}.  For our purposes here, we do not need the full details of the solution, only some general features.  As emphasized in \cite{Chow:2008fe}, there are Killing tensors in string frame, but in general only conformal Killing tensors in other frames, such as Einstein frame.  After analytic continuation, using vielbeins as discussed in Section \ref{liftsol}, the $D$-dimensional string frame metric takes the form
\be
\df s^2 = \sum_{\mu = 1}^n (e^\mu e^\mu + e^{\hat{\mu}} e^{\hat{\mu}}) + \varepsilon e^0 e^0 .
\ee
There is a CCKYT 2-form of the form \eq{Ytilde} \cite{Houri:2010fr}, and the gauge field strength takes the form
\be
F = \fr{c}{s} \sum_{\mu = 1}^n H_\mu e^\mu \wedge e^{\hat{\mu}} .
\ee
$c$, $s$, $H_\mu$ and $x_\mu$ are quantities whose definitions we do not need here.  $F$ takes the form discussed in Section \ref{liftsol}, so the condition \eq{CCKYTliftcond} for the CCKYT 2-form to lift is satisfied.  Its powers $\wtd{Y}^{(j)} = \fr{1}{j!} \wtd{Y} ^{\wedge j}$, for $j = 1, 2, \ldots, \lfloor D/2 \rfloor$ give a tower of CCKYT forms and Hodge dual KYT forms that lift to $D + 1$ dimensions.  If $D \leq 8$, then the solution can be further lifted from $D + 1$ dimensions to 10-dimensional string theory by taking a direct product with a torus, and the Killing tensors trivially lift.

The CCKYT lift condition \eq{CCKYTliftcond} continues to hold for the known AdS generalizations in gauged supergravity, for which only a potential is added to the Lagrangian, in $D = 4$ \cite{Chong:2004na}, $D = 5$ \cite{Chong:2005da}, $D = 6$ \cite{Chow:2008ip} and $D = 7$ \cite{Chow:2007ts}.  When the potential of gauged supergravity is included, the lifted theory does not correspond to a supergravity, so is of less interest.  Nevertheless, the lift condition's other consequences that we discussed in Section \ref{lift} will remain.


\subsection{Five-dimensional black holes}


Five-dimensional $STU$ supergravity consists of minimal $\mathcal N = 2$ supergravity coupled to two vector multiplets.  The bosonic fields are the metric $g_{a b}$, three $\textrm{U}(1)$ gauge fields $A_I$, $I = 1, 2, 3$, and two dilatons $\gvf_i$, $i = 1, 2$.  The Einstein frame bosonic Lagrangian is
\be
\cL_5 = R \star 1 - \frac{1}{2} \sum_{i = 1}^2 \star \df \gvf_i \wedge \df \gvf_i - \frac{1}{2} \sum_{I = 1}^3 X_I^{-2} \star F_I \wedge F_I + F_1 \wedge F_2 \wedge A_3 ,
\ee
where $F_I = \df A_I$ and
\begin{align}
X_1 & = \expe{- \gvf_1/\sr{6} - \gvf_2/\sr{2}} , & X_2 & = \expe{- \gvf_1/\sr{6} + \gvf_2/\sr{2}} , & X_3 & = \expe{2 \gvf_1/\sr{6}} .
\end{align}
Hodge dualizing the gauge field $F_3$ in favour of a 3-form field strength $H$ using $F_3 = - X_3^2 \star H$ gives the Lagrangian
\be
\cL_5 = R \star 1 - \frac{1}{2} \sum_{i = 1}^2 \star \df \gvf_i - \frac{1}{2} \sum_{I = 1}^2 X_I^{-2} \star F_I \wedge F_I - \fr{1}{2} (X_1 X_2)^{-2} \star H \wedge H ,
\ee
where $H$ is given in terms of a 2-form potential $B$ by $H = \df B - \tf{1}{2} (A_1 \wedge F_2 + A_2 \wedge F_1)$.  The 3-charge Cveti\v{c}--Youm solution \cite{Cvetic:1996xz} describes black holes parameterized by a mass, 2 independent angular momenta and 3 independent electric charges.

If we consistently truncate by setting $A_1 = A_2 = A$ and $\gvf_2 = 0$, so that $X_1 = X_2 = X$, then we have the $D = 5$ case of \eq{stringEinstein}.  Minimal $\mathcal N = 2$ supergravity is a further consistent truncation with all gauge fields equal, $A_I = A$, and vanishing dilatons $\gvf_i = 0$.  The 3-charge Cveti\v{c}--Youm solution simplifies in this case.  Following the presentation of \cite{Chow:2008fe}, the metric and gauge field take the form
\begin{align}
\df s^2 & = - e^0 e^0 + \sum_{\mu = 1}^4 e^\mu e^\mu , & A & = \fr{q}{\sr{(r^2 + y^2) R}} e^0 .
\end{align}
The gauge field strength is
\be
F = \fr{2 q}{(r^2 +y^2)^2} (r e^0 \wedge e^1 + y e^2 \wedge e^3) ,
\ee
and its Hodge dual gives
\be
H = \star F = \fr{2 q}{(r^2 + y^2)^2} (y e^0 \wedge e^1 - r e^2 \wedge e^3) \wedge e^4
\ee
A KYT 3-form with torsion $H$ is
\be
Y = (y e^0 \wedge e^1 + r e^2 \wedge e^3) \wedge e^4 ,
\ee
so we are in a Darboux basis for the Hodge dual CCKYT 2-form.  $F$ takes the form discussed in Section \ref{liftsol}, so the lift condition \eq{KYTliftcond} is satisfied, so $Y$ lifts to a 6-dimensional KYT 3-form.  The lift condition continues to hold for the generalization to gauged supergravity \cite{Chong:2005hr}.


\subsection{Four-dimensional black holes}


4-dimensional $STU$ supergravity consists of $\mathcal N = 2$ supergravity coupled to three vector multiplets.  We consider lifts to 5 and 6 dimensions, largely following \cite{Chow:2014cca}.  There is a consistent truncation to one vector multiplet, sometimes called $-\im X^0 X^1$ supergravity.  The bosonic fields are the metric $g_{a b}$, two $\textrm{U}(1)$ gauge fields $A^1$ and $A^2$, a dilaton $\gvf$ and an axion $\chi$.  The Einstein frame Lagrangian, written in terms of $F^1$ and $\wtd{F}_2$, the dual of $F^2$, is
\begin{align}
\cL_4 & = R \star 1 - \tf{1}{2} \star \df \gvf \wedge \df \gvf - \tf{1}{2} \expe{2 \gvf} \star \df \chi \wedge \df \chi - \expe{-\gvf} (\star F^1 \wedge F^1 + \star \wtd{F}_2 \wedge \wtd{F}_2) \nnr
& \quad + \chi (F^1 \wedge F^1 + \wtd{F}_2 \wedge \wtd{F}_2) ,
\end{align}
where $F^I = \df A^I$ and $\df \wtd{F}_I = \df \wtd{A}_I$.  It is a special case of $STU$ supergravity with the gauge fields set pairwise equal.  If we consistently truncate by setting $\wtd{A}_2 = 0$ and dualize the axion $\chi$ in favour of a 3-form field strength $H = - \expe{2 \gvf} \star \df \chi$, then we have the $D = 4$ case of \eq{stringEinstein}.  If we consistently truncate by setting $A^1 = A^2 = A$, $\gvf = 0$ and $\chi = 0$, then we have Einstein--Maxwell theory.

The string frame metric $\df s^2$ and the Einstein frame metric $\df s_{\textrm{E}}^2$ are related by\footnote{This corrects an error in \cite{Chow:2013gba} that led to confusion there.}
\be
\df s^2 = \expe{\gvf} \df s_{\textrm{E}}^2 .
\ee
The 5- and 6-dimensional string frame metrics are
\begin{align}
\df s_5^2 & = \df s^2 + (\df z_1 - A^1)^2 , & \df s_6^2 = \df s_5^2 + (\df z_2 + \wtd{A}_2)^2 .
\end{align}

A charged, rotating black hole solution with both gauge fields dyonic was first found in \cite{LozanoTellechea:1999my}.  Following the notation of \cite{Chow:2013gba}, the string frame metric can be expressed in terms of vielbeins as
\be
\df s^2 = - e^0 e^0 + \sum_{\mu = 1}^3 e^\mu e^\mu ,
\ee
where
\begin{align}
e^0 & = \fr{\sr{(r_2^2 + u_2^2) R}}{W} (\df t + u_1 u_2 \, \df \psi) , & e^1 & = \sr{\fr{r_2^2 + u_2^2}{R}} \, \df r , \nnr
e^2 & = \fr{\sr{(r_2^2 + u_2^2) U}}{W} (\df t - r_1 r_2 \, \df \psi) , & e^3 & = \sr{\fr{r_2^2 + u_2^2}{U}} \, \df u .
\end{align}
The gauge field strengths are
\begin{align}
F^1 & = \fr{[Q_1 (r_2^2 - u_1 u_2) - P^1 u_2 (r_1 + r_2)] e^0 \wedge e^1 + [Q_1 r_2 (u_1 + u_2) + P^1 (r_1 r_2 - u_2^2)] e^2 \wedge e^3}{W (r_2^2 + u_2^2)} , \nnr
F^2 & = \fr{[Q_2 (r_1^2 - u_1 u_2) - P^2 u_1 (r_1 + r_2)] e^0 \wedge e^1 + [Q_2 r_1 (u_1 + u_2) + P^2 (r_1 r_2 - u_1^2)] e^2 \wedge e^3}{W (r_2^2 + u_2^2)} , \nnr
\wtd{F}_1 & = \fr{[P^1 (r_1^2 - u_1 u_2) + Q_1 u_1 (r_1 + r_2)] e^0 \wedge e^1 + [P^1 r_1 (u_1 + u_2) + Q_1 (u_1^2 - r_1 r_2)] e^2 \wedge e^3}{W (r_2^2 + u_2^2)} , \nnr
\wtd{F}_2 & = \fr{[P^2 (r_2^2 - u_1 u_2) + Q_2 u_2 (r_1 + r_2)] e^0 \wedge e^1 + [P^2 r_2 (u_1 + u_2) + Q_2 (u_2^2 - r_1 r_2)] e^2 \wedge e^3}{W (r_2^2 + u_2^2)} .
\end{align}
A string frame KYT 2-form with torsion $H$ is
\be
Y = u_2 e^0 \wedge e^1 + r_2 e^2 \wedge e^3 ,
\ee
so we are in a Darboux basis for $Y$ (and its Hodge dual CCKYT 2-form).  $F$ takes the form discussed in Section \ref{liftsol}, so the lift condition \eq{KYTliftcond} is satisfied under lifting to 5 and 6 dimensions, so $Y$ lifts to a 5-dimensional and 6-dimensional KYT 2-form.  The lift condition continues to hold for the generalization to gauged supergravity \cite{Chow:2013gba}.

Some special cases lifted from 4 to 5 dimensions were considered in \cite{Houri:2012su}.  One example was a lift of the 4-dimensional dyonic Kerr--Newman--NUT solution.  The second example was a 5-dimensional rotating black string \cite{Mahapatra:1993gx}, which from a 4-dimensional perspective corresponds to a rotating black hole where one gauge field, say $F^1$, is electric and the other gauge field, say $F^2$, is magnetic.

An alternative lift from 4 dimensions to 6 dimensions is on $S^2$ \cite{Kerimo:2004md}, leading to a solution of an $\mathcal N = (2, 2)$ gauged supergravity \cite{Kerimo:2003am}.  There is a truncation of the 6-dimensional theory such that the string frame Lagrangian is
\be
\cL_6 = \expe{-\gvf} (R \star 1 + \star \df \gvf \wedge \df \gvf - \tf{1}{2} \star H \wedge H - \textstyle \sum_{I = 1}^3 \star F_I \wedge F_I - 8 g^2 \star 1) .
\ee
For the black hole solution, the 6-dimensional string frame metric is
\be
\df s_6^2 = \df s^2 + \fr{1}{8 g^2} \df \Omega_2^2 ,
\ee
where $\df \Omega_2^2$ is the metric of $S^2$.  This is simply a direct product of the 4-dimensional string frame metric with $S^2$, and so Killing tensors trivially lift.


\section{Conclusion}


We have presented Kaluza--Klein-type lifts of lower-dimensional Killing--Yano forms with torsion to higher dimensions.  The lift condition \eq{KYTliftcond} appears in several different contexts, and it would be interesting to clarify more precisely the connections.  We gave the explicit solution of the lift condition in the generic case.  More general lifts might have more complicated lift conditions.

The Kaluza--Klein ansatz we used for reducing a metric and 3-form was not the most general, as the two lower-dimensional gauge fields were equal and a scalar was truncated.  There are more further Kaluza--Klein reductions to consider, including sphere reductions that give lower-dimensional gauged supergravities.  The theories we considered had a natural 3-form torsion $H$, and it would be interesting to consider more general theories where this is not the case, such as 11-dimensional supergravity with its 4-form field strength.  We focussed on lifting antisymmetric Killing--Yano forms, but there could be further generalizations to symmetric conformal Killing--St\"{a}ckel tensors.


\section*{Acknowledgements}


This work was supported in part by the European Union's Seventh Framework Programme under grant agreements (FP7-REGPOT-2012-2013-1) no.\ 316165, the EU-Greece program ``Thales'' MIS 375734 and was also co-financed by the European Union (European Social Fund, ESF) and Greek national funds through the Operational Program ``Education and Lifelong Learning'' of the National Strategic Reference Framework (NSRF) under ``Funding of proposals that have received a positive evaluation in the 3rd and 4th Call of ERC Grant Schemes''.

\end{document}